\documentclass[12pt]{iopart}
\usepackage{amssymb}
\expandafter\let\csname equation*\endcsname\relax
 \expandafter\let\csname endequation*\endcsname\relax
\usepackage{graphicx}
\usepackage{amsmath}
\usepackage{amsfonts}

\newtheorem{theorem}{Theorem}

\newtheorem{definition}[theorem]{Definition}
\newtheorem{example}[theorem]{Example}

\newtheorem{lemma}[theorem]{Lemma}

\newtheorem{remark}[theorem]{Remark}

\newenvironment{proof}[1][Proof]{\noindent\textbf{#1.} }{\ \rule{0.5em}{0.5em}}
\begin{document}

\title{A measure theoretic approach to linear inverse atmospheric dispersion
problems.}
\author{Niklas Br\"{a}nnstr\"{o}m and Leif \AA\ Persson}
\date{}

\begin{abstract}
Using measure theoretic arguments, we provide a general framework for
describing and studying the general linear inverse dispersion problem where no a
priori assumptions on the source function has been made. We investigate 
the source-sensor relationship
and rigorously state solvability conditions for when the inverse problem can
be solved using a least-squares optimisation method. That is, we derive
conditions for when the least-squares problem is well-defined.
\end{abstract}

\maketitle

\address{Swedish Defence Research Agency, FOI, SE-901 82 Ume\aa, Sweden} 


\section{Introduction}

Atmospheric dispersion models all have the goal to forecast where a
pollutant, if released into the atmosphere, ends up. There are many
applications, e.g. planning where a factory should be located (to reduce the
risk in case of an accident) which requires mainly a local model, or e.g.
forecasting which regions that would be affected by a nuclear power plant
incident (like the Fukushima disaster) which mainly requires a regional or
global model. An equally natural question to ask is: given that we have
detected a pollutant somewhere, can we deduce where the source was located?
If not before, this inverse problem became very important in the wake of the
Chernobyl accident. In that case the radioactive pollution triggered sensors
in Europe before any news of the accident was released. Pinpointing the
location of the source could be done by guessing the location, strength, and
time of the accident and running the dispersion model forward to see whether
it would give the observed measurements. Unless the guess is an educated one
this can be a costly process. The alternative is to solve the inverse
problem. Having a solution to the inverse problem, that is, an estimate of
the parameters in the source function, enables subsequent forward dispersion
modelling to gain a much better understanding of the current state of
affairs (a better situation analysis). Alternatively the source estimate may
be a crucial part of forensic work, for example trying to calculate the
amount of leaked radioactive substances following the accidents in Chernobyl 
\cite{GHL} and Fukushima \cite{StohlEtAl} or pinpointing nuclear test sites 
\cite{RingbomEtAl}.

A number of methods to solve the inverse problem have been suggested. In
addition to the two main contenders Optimisation algorithms and Bayesian
statistics there are methods like Footprint Analysis, e.g. the survey
article \cite{Schmid2002}, Influence Area \cite{Pudykiewicz1998} and \cite%
{Robertson2004}, directly inverting the problem and trying to overcome any
issues associated with ill-conditioning, see e.g. \cite{YF2010}. Often the
methods are designed to bear only on a subclass of inverse dispersion
problems by a priori conditioning on the number of sources, the type of
source, or the dispersion model employed. In the Bayesian approach to the
inverse problem the source is estimated from a so called a posteriori
probability distribution function which is obtained by calculating a
likelihood function and weighing it with any a priori information that one
has at hand (see e.g. \cite{Stuart2010} for an introduction to general
Bayesian inverse problems, and \cite{Franklin1970} for an early reference).
This method avoids the pitfalls of ill-conditioning which are often
associated with directly inverted problems and adds the benefit of allowing
uncertainties in models and measurements to be handled in a tractable
fashion. In a series of papers the Bayesian approach has been adapted to
bear on inverse dispersion problems: in \cite{KYL2007} the case with one
source with unknown position and unknown but constant source strength was
treated. \cite{YF2010} deals with the case where there is a known number of
sources in given locations but where the source strengths are unknown (there
is also an interesting comparison of the results to those obtained with a
directly inverted model where the problems of ill-conditioning have been
alleviated by singular value decomposition). This study was generalised in 
\cite{Yee2007} and \cite{Yee2012} to cover the situation where there is an
unknown number of sources in unknown positions, where the only assumption on
each source is that during emission the source strength is constant. The
case with an unknown number of sources is much harder than working with a
fixed number of sources as the dimension of the parameter space is unknown.
In \cite{Yee2007} and \cite{Yee2012} this problem was overcome by using the
method reversible jump Monte Carlo Markov Chain \cite{Green1995} to sample
from the posterior probability distribution function with an unknown number
of dimensions ( the dimension is one of the parameters that needs to be
estimated). In \cite{Yee2012B} a recursive method is proposed to deal with
same issue.

Under the umbrella of the Optimisation method we find all the various ways
of setting up the inverse dispersion problem so that its solution is given
as the solution of a least-squares fitting problem. As for the Bayesian
method the body of literature mostly covers the case where it is a priori
known that there is only one single source, see e.g. \cite{RL1998}, 
\cite{THG2007}, \cite{AYH2007}, and \cite{ISS2012}. There are exceptions, 
e.g.in \cite{SSI2012} the least-squares method presented in \cite%
{ISS2012} is generalised to cover an unknown number of point sources, and 
in \cite{Bocquet2005} the space-time has been discretised and optimal source 
term is constructed by forming a union of "box-sources" ( the smallest resolution 
is given by the grid box, so "box-source" seems the appropriate term instead of 
point source).

In this paper we are developing a framework for describing
inverse dispersion problems. The framework relies on using measure theoretic 
ideas and methods to study the general linear inverse dispersion problem without making a
priori assumptions on e.g. the number of sources, their emission patterns or
their distribution in the spatio-temporal domain. As such, the framework is non-parametric 
but since the term non-parametric seems to be overloaded we refrain from using it to describe 
the framework. 
We begin by setting up the linear inverse problem and then we present a  one-dimensional 
toy problem that motivate the use of measures rather than probability densities. 
Then we turn to the problem of determining under which conditions a given set of sensor 
data can be generated by a source chosen from a given class of sources. As a warm-up 
we consider linear combinations of base source measures in  both the invertible case 
and the over determined case. The arguments are based on finding appropriate cones 
in the space of positive measures (describing the source) and in the space of sensor 
measurements.We then build on this to generalise the analysis to 
the case where the source is chosen from a closed cone of measures (we dispense 
of the assumption of having a finite number of base sources), Theorem 
\ref{thm:CompactCondition2}. While Theorem \ref{thm:CompactCondition2} 
is certainly interesting in its own right explaining when a measurement 
can be realised the analysis also allows for a derivation of the main result of 
the paper: conditions under which the least squares optimisation problem is 
well-defined, Theorem \ref{thm:LeastSquare}. In addition to these results 
we also characterise the set of measurements when a source is approximated 
by a sequence of instantaneous point sources, 
Theorem \ref{thm:ConvexConicalHull}.

The measure theoretic approach that is presented in this paper introduces a
machinery which we believe will be useful in future studies where rigorous
results on general linear inverse dispersion problem are sought. Indeed, while not
solving any particular inverse dispersion problem, the method is not
hampered by any peculiarities that a given set of parameters could have
introduced.

\section{Setting of the problem, the dispersion model and its adjoint}

The atmospheric dispersion problem that we are interested in can be
formulated in terms of a transition probability $p(t,x;t^*,x^*)$, where 
$(t^*,x^*),(t,x)\in T\times V$ where $T\subset \mathbb{R}$ is a time interval and 
$V\subset \mathbb{R}^{3}$ is a spatial domain. The transition probability
density expresses the probability for a particle released at the time-space
point $(t^*,x^*)$ to reside in the time-space point $(t,x)$ for $t\geq t^*$. We
note that $p=0$ when $t<t^*$. The particles whose dispersion is governed by
this transition probability is assumed to originate from a source $S$. The
source $S$ is assumed to be a positive measure on $T\times V$ (that is, the word 
"source" is used in the strict sense; no sinks are considered in this paper). In this way
the total mass\ $M$ released from the source is given by integrating the
source measure $S$ over its support%
\begin{equation}
M=\int_{T}\int_{V}dS(t^*,x^*).
\end{equation}%
The quantity that is usually desired as output from a dispersion model is
the concentration of the pollutant in a given space-time point. Since $S$
has its support on $T\times V$ and the transition probability describes the
dynamics of the released substance the concentration $c(t,x)$ is obtained by
weighing all released particles (released at some $(t^*,x^*)$ with $t^*<t$) with
the probability that they have been transported from $(t^*,x^*)$ to $(t,x)$%
\begin{equation}
c(t,x)=\int_{T}\int_{V}p(t,x;t^*,x^*)dS(t^*,x^*).
\end{equation}%
While $c(t,x)$ is the predicted concentration at the space time point $(t,x)$
the sensor may not have the resolution to make an ideal measurement from the
concentration field $c(t,x)$, indeed the sensor may perform some form of
averaging in both space and time to yield the sensor response $\overline{c}%
(t,x)$. We assume that the averaging process in the sensor can be described
by a probability measure $S^{\ast }$ (usually referred to as the
sensor-filter function) on $T\times V$, and hence we express the sensor
response as%
\begin{equation}
\overline{c}=\int_{T}\int_{V}c(t,x)dS^{\ast }(t,x).
\end{equation}%
According to Fubini's theorem (\emph{Rudin, Theorem 8.8 p. 164}), this can
be written as%
\[
\bar{c}=\int \int cdS^{\ast }=\int \int \int \int pd\left( S\times S^{\ast
}\right) =\int \int c^{\ast }dS
\]%
In case $S$ and $S^{\ast }$ are given by square--integrable spacetime
densities $dS\left( t,x\right) =s\left( t,x\right) dtdx$, $dS^{\ast
}=s^{\ast }\left( t,x\right) dtdx$, then $c$ and $c^{\ast }$ are also
square--integrable spacetime densities, and 
\begin{equation}
\bar{c}=\left( c,s^{\ast }\right) =\left( c^{\ast },s\right) 
\label{eqn:adjointfcns}
\end{equation}%
where the inner product is defined by $\left( f,g\right) =\int \int f\left(
t,x\right) g\left( t,x\right) dtdx$. Therefore, $c^{\ast }$ is called the 
\emph{adjoint concentration}. We want to allow sources and measurements with
singular parts. Let us consider combinations of square--integrable spacetime
densities and instantaneous point masses. To this end we generalize $c$ and $%
s$ to measures of the form

\[
dF\left( t^{\ast },x^{\ast }\right) =f\left( t^{\ast },x^{\ast }\right)
dt^{\ast }dx^{\ast }+\sum_{j}f_{j}\delta _{\left( t_{j}^{\ast },x_{j}^{\ast
}\right) }
\]%
which we call \emph{primal measures}, and we generalize $c^{\ast }$ and $%
s^{\ast }$ to measures of the form 
\[
dG^{\ast }\left( t,x\right) =g^{\ast }\left( t,x\right)
dtdx+\sum_{i}g_{i}^{\ast }\delta _{\left( t_{j},x_{j}\right) }
\]%
which we call \emph{dual measures}.  We would like to have a generalization
of (\ref{eqn:adjointfcns}) to  
\begin{equation}
\bar{c}=\left\langle C,S^{\ast }\right\rangle =\left\langle S,C^{\ast
}\right\rangle   \label{eqn:adjointmeasures}
\end{equation}%
for a suitable bilinear map $\left\langle \cdot ,\cdot \right\rangle $,
which implies that 
\[
\left\langle F,G^{\ast }\right\rangle =\left( f,g^{\ast }\right)
+\sum_{j}f\left( t_{j},x_{j}\right) g_{j}^{\ast }+\sum_{i}f_{i}g^{\ast
}\left( t_{i}^{\ast },x_{i}^{\ast }\right) 
\]%
This definition makes sense only if

\begin{itemize}
\item $f$ is continuous at $\left( t_{j},x_{j}\right) $ so the application
of $\delta _{\left( t_{j},x_{j}\right) }$ in $G^{\ast }$ is appropriate

\item $g^{\ast }$ is continuous at $\left( t_{i}^{\ast },x_{i}^{\ast
}\right) $ so the application of $\delta _{\left( t_{i}^{\ast },x_{i}^{\ast
}\right) }$ in $F$ is appropriate

\item The $\left( t_{j},x_{j}\right) $'s are disjoint from the $\left(
t_{i}^{\ast },x_{i}^{\ast }\right) $'s, because multiplication of point
masses with common support is not defined. In other words, we are not
allowed to make an instantaneous point measurements at the spacetime
location of an instantaneous point source.
\end{itemize}

In fact, in case $F=C$ and $G^{\ast }=S^{\ast }$ the two first conditions
implies the third, because of the connection between $S$ and $C$. Indeed,
since $S$ is a primal measure we have%
\begin{eqnarray*}
c\left( t,x\right) &=&\int \int p\left( t,x;t^{\ast },x^{\ast }\right)
s\left( t^{\ast },x^{\ast }\right) dt^{\ast }dx^{\ast }+\sum_{j}p\left(
t,x;t_{j}^{\ast },x_{j}^{\ast }\right) s_{j} \\
c_{i} &=&0
\end{eqnarray*}%
so $C$ has no singular part (point masses) and the density $c\left(
t,x\right) $ has singularities at $\left( t_{j}^{\ast },x_{j}^{\ast }\right) 
$ and hence cannot be continuous there. By a similar argument, $C^{\ast }$
cannot have any singular part.

To summarize, we have a generalization (\ref{eqn:adjointmeasures}) of the
adjoint to measures, with the inner product $\left( \cdot ,\cdot \right) $,
a bilinear form, is replaced by a bilinear mapping $\left\langle \cdot
,\cdot \right\rangle $ on the spaces of primal and adjoint measures. The
common part is the subspace of square--integrable densities, on which the
bilinear mapping $\left\langle fdxdt,g^{\ast }dxdt\right\rangle $ coincides
with the inner product $\left( f,g^{\ast }\right) $. Since the measures $C$
and $C^{\ast }$ do not have singular parts, we henceforth identify them with
their densities $c$ and $c^{\ast }$, and refer to $c^{\ast }$ as the adjoint
concentration, although it is not an adjoint in the usual Hilbert space
sense. Similar formulas apply for other combinations of singular measures
like continuous or instantaneous point, line, area or volume sources or
measurements, but is not elaborated on further here.

Let us now use the definition of $c(t,x)$ to rewrite this expression in the
following way%
\begin{eqnarray}
\overline{c} &=&\int_{T}\int_{V}c(t,x)dS^{\ast }(t,x)  \notag \\
&=&\int_{T}\int_{V}\int_{T}\int_{V}p(t,x;t^*,x^*)dS(t^*,x^*)dS^{\ast }(t,x) \\
&=&\int_{T}\int_{V}\left( \int_{T}\int_{V}p(t,x;t^*,x^*)dS^{\ast }(t,x)\right)
dS(t^*,x^*).  \notag
\end{eqnarray}%
By defining the \emph{adjoint concentration field }$c^{\ast }(t^*,x^*)$ as%
\begin{equation}
c^{\ast }(t^*,x^*)=\int_{T}\int_{V}p(t,x;t^*,x^*)dS^{\ast }(t,x)  \label{def_adjoint}
\end{equation}%
we get%
\begin{equation}
\overline{c}=\int_{T}\int_{V}c^{\ast }(t^*,x^*)dS(t^*,x^*).
\end{equation}%
Hence we have two equivalent ways of calculating the sensor response%
\begin{equation}
\overline{c}=\int_{T}\int_{V}c(t,x)dS^{\ast }(t,x)=\int_{T}\int_{V}c^{\ast
}(t^*,x^*)dS(t^*,x^*)
\end{equation}%
which is the dual relationship between the forward and the adjoint
description of the dispersion problem. We note that equation (\ref%
{def_adjoint}) describing the adjoint concentration field is evolving
backwards in time: we may view the transition probability as moving adjoint
particles released by $S^{\ast }$ backwards in time and space. The main
advantage of using the adjoint representation in inverse dispersion
modelling is computational efficiency. This is a well-documented fact, see
for example \cite{Marchuk1986}. We also remark that the adjoint
concentration field $c^{\ast }$ is independent of the source function $S$,
and the concentration field $c$ is independent of the sensor-filter function 
$S^{\ast }$.

\section{Source-receptor relationship}

The dispersion problem predicts how a pollutant from a source spreads in the
atmosphere. From an abstract point of view this problem can be seen as a
problem of mapping of measures: the source $S$ can be viewed as a measure in
the spatio-temporal domain $T\times V$ that is being mapped via the
dispersion equations into a scalar function $c$ (the concentration), from
which we make measurements represented by a probability measure $S^{\ast }$,
defining the averaging of the concentration function $c$. From this level of
abstraction the adjoint version of the problem is very similar. In this case
the adjoint equations maps a probability measure $S^{\ast }$ on $T\times V$
representing a measurement in a sensor to a scalar function $c^{\ast }$
(adjoint 'concentration') from which we can make "adjoint measurements"
using a source measure $S$ acting on the adjoint 'concentration' $c^{\ast }$%
. (Depending on the scaling of the problem the adjoint 'concentration' $%
c^{\ast }$ may not be a proper concentration dimensionally.) In view of this
light, asking questions about the sensor response in the forward problem or
asking questions about the source in the inverse problem are very similar.
Based on this observation we therefore propose to adopt a measure theoretic
approach and we develop a mathematical framework for studying the inverse
problem. While we are omitting the analysis of the forward problem in this
paper we note that treating this problem is completely analogous. Studying
the problem in this generality will not allow us to solve any particular
inverse dispersion problem, but it will allow us to draw general conclusions
about whole classes of problems. One particular advantage of this approach
hence lies in the fact that we avoid difficulties that may be associated
with a particular problem and its parameters - of course, these will have to
be addressed when the particular problem is to be solved.

\section{One--dimensional example to motivate the use of measure theory}

As a model example, consider a stationary one--dimensional diffusion on the
unit interval with absorbing boundary conditions%
\begin{eqnarray}
-c^{\prime \prime }\left( x\right) &=&S\left( x\right) \text{, }x\in \left[
0,1\right]  \label{eqn:1DproblemDE} \\
c\left( 0\right) &=&c\left( 1\right) =0  \label{eqn:1DProblemBC}
\end{eqnarray}%
The solution $c\left( x\right) $ is a concave function; using the integral
formula of Blaschke and Pick \cite{BlaschkePick1916} the solution can be
written%
\[
c\left( x\right) =\int_{0}^{1}\frac{y\left( 1-y\right) }{\sqrt{3}}\hat{%
\varphi}\left( x,y\right) S\left( y\right) dy 
\]%
where\footnote{%
The basis functions $\hat{\varphi}$ are normalized so that $\int_{0}^{1}\hat{%
\varphi}^{2}\left( x,y\right) dx=1$}%
\[
\hat{\varphi}\left( x,y\right) =\left\{ 
\begin{array}{ccc}
\sqrt{3}x/y & \text{if} & 0\leq x\leq y \\ 
\sqrt{3}\left( 1-x\right) /\left( 1-y\right) & \text{if} & y\leq x\leq 1%
\end{array}%
\right. 
\]%
This formula is also valid if $S$ is a unit point mass at a fixed point $y$,
in which case the concentration profile is%
\[
c\left( x\right) =\frac{y\left( 1-y\right) }{\sqrt{3}}\hat{\varphi}\left(
x,y\right) =\min \left( x\left( 1-y\right) ,y\left( 1-x\right) \right)
\equiv f\left( x,y\right) 
\]%
and in case of a point measurement at a fixed point $x$ we have $c^{\ast
}\left( y\right) =f\left( x,y\right) $. Given a finite number of measurement
points $x_{1},...,x_{m}$ a vector of measured values $\bar{c}=\left(
c_{1},...,c_{m}\right) $ is the result of a smooth density $S\left( y\right) 
$ if and only if the points $\left( 0,0\right) ,\left( x_{1},c_{1}\right)
,...,\left( x_{m},c_{m}\right) ,\left( 1,0\right) $ lie on the graph of the
smooth concave function $c\left( x\right) $ given by the formula above.
Likewise, $\bar{c}$ is the result of a point source $S$ at $y$ if and only
if the same points lie on the graph of a function $\lambda f\left( \cdot
,y\right) $ for some $\lambda >0$.

We want the set of measurement vectors $\bar{c}$ to be closed, so that we
can determine the closest measurement vector from any given vector. Taking a
sequence $c_{j}$ of smooth convex functions converging pointwise towards $%
f\left( \cdot ,x_{k}\right) $ for some $1<k<m$ we conclude that the vector $%
\bar{c}=\left( f\left( x_{1},x_{k}\right) ,...,f\left( x_{m},x_{k}\right)
\right) $ should be included. The points $\left( 0,0\right) ,\left( x_{1},%
\bar{c}_{1}\right) ,...,\left( x_{k},\bar{c}_{k}\right) $ are collinear, and
likewise $\left( x_{k},\bar{c}_{k}\right) ,...,\left( x_{m},\bar{c}%
_{m}\right) ,\left( 1,0\right) $, and the only concave function containing
these points in its graph is $f\left( \cdot ,x_{k}\right) $ so $\bar{c}$
must come from a point source at $x_{k}$. Hence point sources must be
allowed. Since any measure can be locally approximated (by weak convergence
of measures) by a sequence of finite linear combinations of point sources,
it is natural to allow sources given by finite measures.

\section{Linear combinations of sources}

The purpose of this section is to characterize all possible measurement
values obtainable when $S$ is a linear combination of a given finite number
of base sources. In other words, we \ will now investigate under which
condition there exists a measure $S$ which will produce the concentration
measurements exactly. Finding a source $S$ reproducing a value $\bar{c}$ for a
measurement $S^{\ast }$ is easy; simply take an arbitrary source that gives
a positive measured value and scale the source properly. Trying the same
idea for several measurements $S_{i}^{\ast }$, $i=1,...,m$, take sources $%
S_{j}$, $j=1,...,n$ and assume that $S=\sum_{j=1}^{n}\lambda _{j}S_{j}$ with 
$\lambda _{j}\geq 0$ (we only consider $\lambda _{j}\geq 0$ since we want
all $S_{j}$ to contribute as sources, were some $\lambda _{j}$ allowed to be
negative the corresponding "source" $S_{j}$ would act as a sink, even if $S$
could still be positive). Given the measured values $\bar{c}_{1},...,\bar{c}%
_{m}\geq 0$ we get the linear system of equations%
\begin{equation}
\sum_{j=1}^{n}a_{ij}\lambda _{j}=\bar{c}_{i}\text{ where }%
a_{ij}=\left\langle S_{j},c_{i}^{\ast }\right\rangle 
\end{equation}%
and we denote $A=\left( a_{ij}\right) $, which is sometimes called the \emph{%
source--receptor matrix}. Assume first that $A$ is invertible, i.e., $m=n$
and the measurement vectors $\left( \left\langle S_{j},c_{1}^{\ast
}\right\rangle ,...,\left\langle S_{j},c_{m}^{\ast }\right\rangle \right) $
(produced by the individual sources $S_{j}$, $j=1,...,n$) are linearly
independent. Then, since $A$ is an invertible nonnegative matrix (by
nonnegative matrix we mean a matrix where all elements are nonnegative), the
inverse $A^{-1}$ contains nonpositive elements on row $i$ if $A$ contains
off--diagonal positive elements in column $i$ (see the remark below for justification). 
Hence the condition that $%
\lambda _{i}\geq 0$ gives a linear constraint 
\begin{equation}
-\sum_{j\in J_{i}^{-}}\left( a^{-1}\right) _{ij}\bar{c}_{j}\leq \sum_{j\in
J_{i}^{+}}\left( a^{-1}\right) _{ij}\bar{c}_{j}
\end{equation}%
where $J_{i}^{+}$ denotes the set of column indices $j$ for which $\left(
a^{-1}\right) _{ij}>0$ and $J_{i}^{-}$ denotes the set of column indices $j$
for which $-(a^{-1})_{ij}>0$.

\begin{remark}
Suppose that $A\geq 0$. The row vectors $A_{i}^{-1}$ of $A^{-1}$ and the
column vectors $A_{j}$ of $A$ satisfy $A_{i}^{-1}\cdot A_{j}=\delta _{ij}$.
Suppose that $A_{i}$ contains $k$ positive components, e.g., $A_{i}=\alpha
_{1}e_{1}+...+\alpha _{k}e_{k}$ with $\alpha _{l}>0$, $l=1,...,k$. If $j\neq
i$ and $A_{j}^{-1}=\beta _{1}e_{1}+...+\beta _{n}e_{n}$ then either $\beta
_{1}=...=\beta _{k}=0$ or $\beta _{l}<0$ for some $1\leq l\leq k$. In the
former case we have $A_{j}^{-1}=\beta _{k+1}e_{k+1}+...+\beta _{n}e_{n}$.
There can be at most $n-k$ such $A_{j}^{-1}$'s since the $A_{j}^{-1}$'s are
linearly independent. Hence there are at least $k-1$ column vectors $%
A_{j}^{-1}$'s with $j\neq i$ that contain negative elements. Therefore, both
the positive and negative parts $\left( A^{-1}\right) _{ji}^{+}=\max \left(
0,\left( A^{-1}\right) _{ji}\right) $ and $\left( A^{-1}\right)
_{ji}^{-}=\max \left( 0,-\left( A^{-1}\right) _{ji}\right) $ are nonzero,
and the nonnegativity conditions $\lambda =A^{-1}\bar{c}\geq 0$ give the
linear constraints%
\[
\left( A^{-1}\right) ^{-}\bar{c}\leq \left( A^{-1}\right) ^{+}\bar{c}
\]
\end{remark}

The general case requires more work, but may be solved as a minimization problem, 
indeed the problem of finding the "best" nonnegative solution $x\in \mathbb{R}^{n}$%
, $x\geq 0$ to the linear system $Ax=b$, where $A\in \mathbb{R}^{m\times n}$
and $b\in \mathbb{R}^{m}$ are given, can be formulated as a constrained
quadratic minimization problem%
\begin{eqnarray*}
\min d\left( x,z\right) &=&\frac{1}{2}\left\Vert z^{2}\right\Vert \\
Ax-b-z &=&0 \\
x &\in &\mathbb{R}^{n}\text{, }z\in \mathbb{R}^{m}\text{, }x\geq 0
\end{eqnarray*}%
The Lagrangian for this problem is%
\[
\mathcal{L}\left( x,z,\mu ,\eta \right) =d\left( x,z\right) +\mu
^{T}\left( Ax-b-z\right) -\eta ^{T}x, 
\]%
where $\cdot^T$ denotes the transpose.
Necessary and sufficient conditions for optimal points $\left( x,z\right) $
are given by the Karush--Kuhn--Tucker conditions: $\left( \hat{x},\hat{z}
\right) $ is an optimal point if and only if there are $\hat{\mu}\in 
\mathbb{R}^{m}$, $\hat{\eta} \in \mathbb{R}^{n}$ such that%
\begin{eqnarray}
\nabla _{x}\mathcal{L} &\mathcal{=}&A^{T}\hat{\mu}-\hat{\eta}=0
\label{eqn:KKTstationarityX} \\
\nabla _{z}\mathcal{L} &\mathcal{=}&\hat{z}-\hat{\mu}=0
\label{eqn:KKTstationarityZ} \\
A\hat{x}-b-\hat{z} &=&0\text{ }  \label{eqn:KKTPrimalFeasibilityAx} \\
\hat{x} &\geq &0\text{ }  \label{eqn:KKTPrimalFeasibilityX} \\
\hat{\eta} &\geq &0\text{ }  \label{eqn:KKTDualFeasibilityLambda1} \\
\hat{\eta _{j}}\hat{x_{j}} &=&0\text{, }j=1,...,n\text{ }
\label{eqn:KKTComplementarityLambda}
\end{eqnarray}%
This system can be solved by the linear program%
\begin{eqnarray}
\min w &=&\Sigma _{i}u_{i}+\Sigma _{j}v_{j} \\
A^{T}\mu -\eta &=&u \\
Ax-b-\mu &=&v\text{ } \\
x,\eta ,u &\in &\mathbb{R}^{n}\text{, }z,\mu ,v\in \mathbb{R}^{m}\text{, }%
x,\eta ,u,v\geq 0
\end{eqnarray}%
using a modification of the simplex method, where the complementarity
conditions (\ref{eqn:KKTComplementarityLambda}) are enforced by a restricted
basis entry rule (\emph{suitable reference inserted here...}). The
Lagrangian dual objective function%
\[
q\left( \mu ,\eta \right) =\inf_{x\in \mathbb{R}^{n},z\in \mathbb{R}^{m}}%
\mathcal{L}\left( x,z,\mu ,\eta \right) 
\]%
is defined on%
\[
\mathcal{D=}\left\{ \left( \mu ,\eta \right) :\mu ^{T}A-\eta
=0\right\} \text{,} 
\]%
and for $\left( \mu ,\eta \right) \in \mathcal{D}$ the minimum occurs for 
$x\in \mathbb{R}^{n}$, $z=\mu $ which gives 
\[
q\left( \mu ,\eta \right) =\mathcal{L}\left( x,\mu ,\mu ,\eta \right)
=-\mu ^{T}b-\frac{1}{2}\left\Vert \mu \right\Vert ^{2}\text{.} 
\]%
The Lagrangian dual problem is%
\begin{eqnarray*}
\max q\left( \mu ,\eta \right) &=&-\mu ^{T}b-\frac{1}{2}\left\Vert \mu
\right\Vert ^{2} \\
\mu ^{T}A-\eta &=&0 \\
\mu &\in &\mathbb{R}^{m},\eta \in \mathbb{R}^{n},\eta \geq 0
\end{eqnarray*}%
Since $d\left( x,z\right) $ is convex we have strong duality, i.e., $\max
q\left( \mu ,\eta \right) =q( \hat{\mu},\hat{\lambda})
=d\left( \hat{x},\hat{z}\right) =\min d\left( x,z\right) $ which by the
KKT conditions gives the optimal value%
\[
-\hat{\mu} ^{T}b-\left\Vert \hat{\mu}\right\Vert ^{2}/2=q=d=\left\Vert
\hat{z}\right\Vert ^{2}/2 
\]
There are two mutually exclusive cases: either the optimal value is $0$ (in
which case $\hat{\mu}=\hat{z}=0$ and $Ax=b$ has a solution $x\geq 0$)
or the optimal value is $>0$ (in which case $\hat{\mu}=\hat{z}$, 
$\left\Vert \hat{z}\right\Vert ^{2}=-\hat{z}^{T}b>0$ and $Ax=b$ does not
have a solution $x\geq 0$).

Considering the directional derivative of $q\left( \mu ,\eta \right) $ at 
$\left( 0,\eta \right) $ in the feasible direction $\nu $ ($\nu ^{T}A\geq
0$) we see that%
\[
\left. \frac{d}{dt}q\left( t\nu ,\eta \right) \right\vert _{t=0}=-\nu
^{T}b 
\]%
so $\mu =0$ is optimal in the dual problem (i.e., $Ax=b$ has a solution $%
x\geq 0$) if and only if $-\nu ^{T}b\leq 0$ for all feasible $\nu $ ($\nu
^{T}A\geq 0$). This is the content of the famous \emph{Farkas' lemma}, 
see e.g. \cite{Franklin1980}, p. 56.

The simplex method and Farkas' lemma have an instructive geometrical 
interpretation: the column vectors of $A$ generates a polyhedral cone 
$\kappa_{A}=\left\{ b\in \mathbb{R}^{n}:\exists x\in \mathbb{R}_{+}^{m}\text{ and 
}Ax=b\right\}$. If the optimal value is 0 then $b$ belongs to the cone $\kappa_{A}$ 
and we can find $x$ such that $Ax=b$, while if the optimal value is $>0$ then $b$ 
lies outside the cone and the optimal solution $\hat{x}$ is the point on the boundary 
of the cone minimizing the "distance" $d$ between $Ax$ and $b$, see Figure \ref{fig_simplex}.

When the optimal value is zero it means that the measurements $\Bar{c}$ 
can be realised exactly by a linear combination of the given sources, and 
hence the cone represents all possible measurements obtainable by linearly 
combining the given base sources.

A linear combination of Dirac measures is particularly interesting since
these are extremal elements in the convex sense, and if $S=\delta _{t^*,x^*}$,
then $\left\langle S,c^{\ast }\right\rangle =c^{\ast }\left(t^*,x^*\right) $,
so the measurement values obtained from a linear combination of Dirac
measures consist of the polyhedral cone generated by the values of $c^{\ast
} $ at the support points of the Dirac measures.\ This is generalized to
arbitrary positive measures below, see Theorem \ref{thm:ConvexConicalHull}.

\begin{figure}[!ht]
\centering
\includegraphics[scale=0.45]{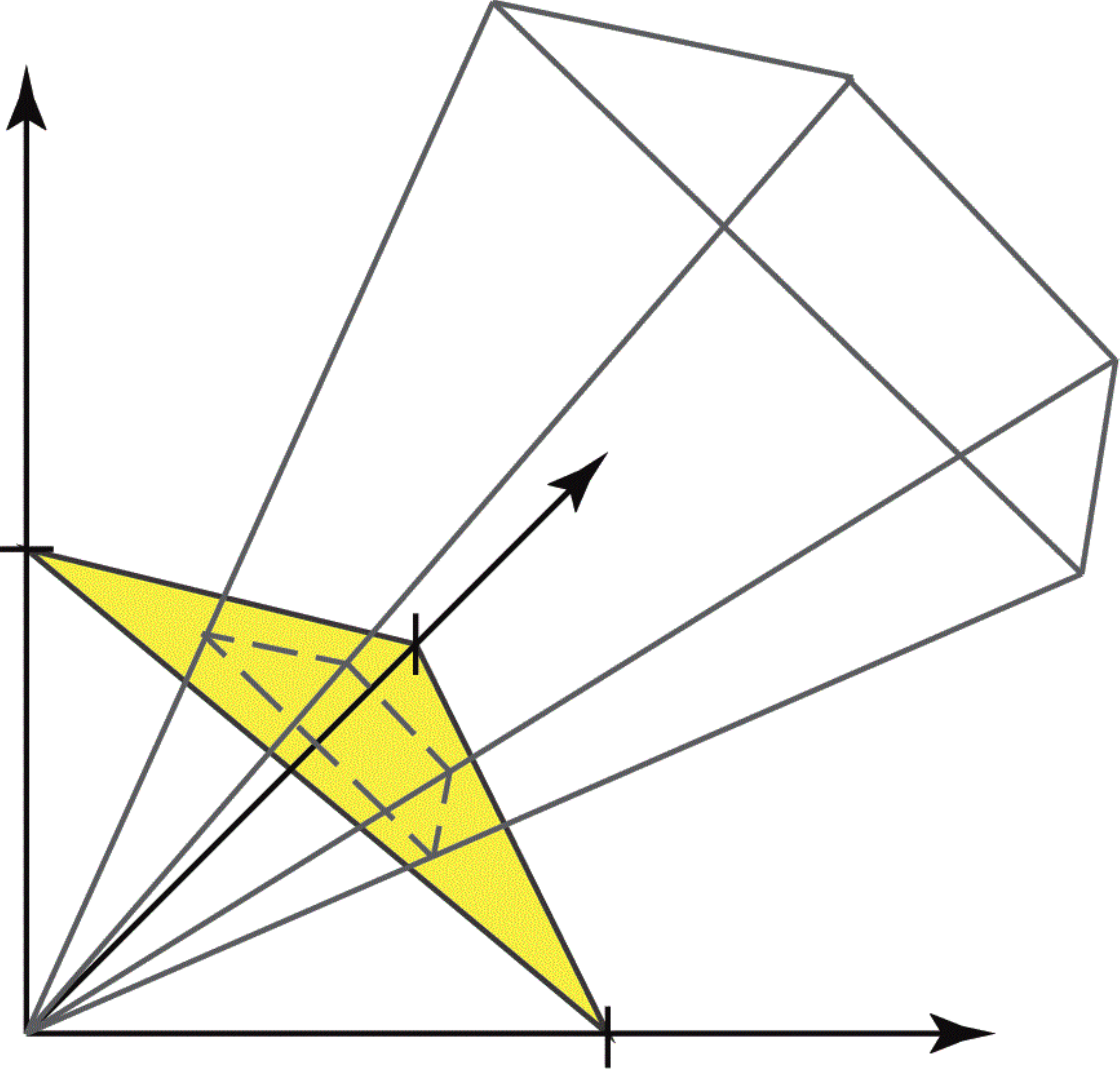}
\caption{The standard simplex in $\mathbb{R}^{3}$ and its intersection with
a cone in the positive octant.}
\label{fig_simplex}
\end{figure}

%

\section{Measurements of arbitrary sources}

The purpose of this section is to characterize all possible measurement
values obtainable when $S$ is picked from a more general closed cone of
positive measures. For this purpose, we define the measurement operator with
respect to the given adjoint function $c^{\ast }$:

\begin{definition}
\label{def:Tcstar}Given a set $\mathbb{S}$ of positive measures and
nonnegative continuous functions $c^{\ast }=c_{1}^{\ast },...,c_{m}^{\ast }$
on $T\times V$ we define 
\begin{equation}
H_{c^{\ast }}\left( S\right) =\left( \left\langle S,c_{1}^{\ast
}\right\rangle ,...,\left\langle S,c_{m}^{\ast }\right\rangle \right) \in 
\mathbb{R}_{+}^{m}
\end{equation}%
for all $S\in \mathbb{S}.$
\end{definition}

The results in the previous section shows that if $\mathbb{S}$ is a finite
positive cone (generated by the given sources $S_{1},...,S_{n}$) then the
image $H_{c^{\ast }}\left( \mathbb{S}\right) $ is a polyhedral cone in $%
\mathbb{R}_{+}^{m}$. In this section we drop the assumption on having 
a fintite number of base sources and investigate whether we still can 
draw similar conslusions about the measurement
values (i.e., the image of $H_{c^{\ast }}$).

We need to impose some structure (restrictions) on the set of source 
measures to perform the analysis, in particular we will make use of the 
notions of tightness and compactness.

\begin{definition}
\label{def:UniformTightness}A set of positive measures $\mathbb{S}$ on $%
T\times V$ is said to be \emph{uniformly} \emph{tight} if for each $%
\varepsilon >0$ there is a compact set $K_{\varepsilon }\in T\times V$ such
that $S\left( K_{\varepsilon }^{c}\right) <\varepsilon $ for all $S\in 
\mathbb{S}$, where the set $K_{\varepsilon }^{c}$ denotes the complement of $%
K_{\varepsilon }$ in $T\times V$.
\end{definition}

Loosly speaking the definition says that the mass\footnote{here the word mass 
refers to the value of the measure on the set, not physical mass} contained in the complement 
of the compact set $K_{\epsilon}$ can be made arbitrarily small, that is, nearly 
all mass is contained in the compact set  $K_{\epsilon}$, which intuitively means 
that the conceivable sources are not allowed to
release "too much mass too far away and too long ago". Measures can be
constructed with approximation methods, and to show that approximations
converges to the sought solution, we need appropriate compactness
properties, in this case the following:

\begin{definition}
\label{def:WeakRelativeCompactness}A set of positive measures $\mathbb{S}$
on $T\times V$ is said to be \emph{weakly relatively compact} if for any
sequence of measures $\left( S_{j}\right) _{j=1}^{\infty }$ in $\mathbb{S}$
there is a subsequence $j_{k}\rightarrow \infty $ when $k\rightarrow \infty $
such that $S_{j_{k}}$ is weakly convergent when $k\rightarrow \infty $,
i.e., there is a measure $S$ (not necessarily in $\mathbb{S}$, unless $%
\mathbb{S}$ is weakly closed) such that $\int fdS_{j_{k}}\rightarrow \int fdS
$ when $k\rightarrow \infty $, for all bounded continuous functions $f$.
\end{definition}

The notion of compactness and tighness are related:

\begin{theorem}
\label{thm:Prohorov}(Prohorov's theorem) A set of positive measures $\mathbb{%
S}$ is weakly relatively compact if and only if $\mathbb{S}$ is uniformly
tight and all $S\in \mathbb{S}$ have uniformly bounded total masses.
\end{theorem}

\begin{proof}
See \cite{DaleyVere-Jones2003}, p. 394--396.
\end{proof}

\begin{example}
The tightness condition is necessary. Consider a sequence of Dirac measures $%
S_{j}$ at discrete spacetime points $\left( t_{j},x_{j}\right) $ converging
to infinity. Then the $S_{j}$'s are not tight since any compact set is
eventually avoided by $\left( t_{j},x_{j}\right) $, but all $S_{j}$'s have
mass $1$ and hence uniformly bounded. There can be no weakly convergent
subsequence $S_{j_{k}}$, since that would mean that $f\left(
t_{j_{k}},x_{j_{k}}\right) $ is convergent for all continuous functions $f$.
\end{example}

We are now in the position to show that if we consider source measures 
in the weak closure of a set of measures that are uniformly tight and has 
uniformly bounded total masses then the
attainable measurement values $H_{c^{\ast }}\left( \overline{\mathbb{S%
}}\right)$ constitutes a closed and bounded set.

\begin{lemma}
\label{thm:CompactCondition}If $\mathbb{S}$ is uniformly tight and has
uniformly bounded total masses, then $H_{c^{\ast }}\left( \overline{\mathbb{S%
}}\right) $ is a compact subset of $\mathbb{R}_{+}^{m}$, where $\overline{%
\mathbb{S}}$ denotes closure of $\mathbb{S}$ with respect to weak
convergence of measures.
\end{lemma}

\begin{proof}
Assume that $y_{j}\in H_{c^{\ast }}\left( \mathbb{S}\right) $, i.e., there
are measures $S_{j}$ such that $y_{j}=H_{c^{\ast }}\left( S_{j}\right) $,
and assume that $y_{j}\rightarrow y$ when $j\rightarrow \infty $. By
Prohorov's theorem, there is a subsequence $j_{k}\rightarrow \infty $ when $%
k\rightarrow \infty $ and a measure $S\in \overline{\mathbb{S}}$ such that $%
S_{j_{k}}\rightarrow S$ weakly when $k\rightarrow \infty $, which implies
that $H_{c^{\ast }}\left( S_{j_{k}}\right) \rightarrow H_{c^{\ast }}\left(
S\right) $ when $k\rightarrow \infty $. Hence $H_{c^{\ast }}\left( S\right)
=y$, so $H_{c^{\ast }}\left( \overline{\mathbb{S}}\right) $ is closed.
Moreover, $H_{c^{\ast }}\left( \overline{\mathbb{S}}\right) $ is bounded
since $c^{\ast }$ is bounded and $\overline{\mathbb{S}}$ has uniformly
bounded total masses.
\end{proof}

\begin{example}
(Single instantaneous point sources) Let $D\subset T\times V$ be an open
subset of the spacetime domain, and let $\mathbb{S}$ be the set of
instantaneous point sources in $D$ with mass $M>0$. Then $\overline{\mathbb{S%
}}$ is the set of instantaneous point sources in $\overline{D}$ (the closure
of $D$ in $T\times V$) with mass $M$, and $H_{c^{\ast }}\left( \overline{%
\mathbb{S}}\right) =\left\{ Mc^{\ast }\left( t,x\right) :\left( t,x\right)
\in \overline{D}\right\} $. Hence the attainable measurement values for $%
\mathbb{S}$ is a surface in $\mathbb{R}_{+}^{m}$ parametrized over the
four--dimensional domain $\overline{D}$.
\end{example}

In Lemma \ref{thm:CompactCondition} tightness and uniformly bounded 
total masses implies compactness of  $H_{c^{\ast }}\left( \overline{\mathbb{S%
}}\right)$, in order to sharpen this statement by replacing the implication by 
equivalence we introduce a particular kind of tightness adapted to $c^{\ast }$, 
indeed we consider compact sets constructed from level sets of  $c^{\ast }$.

\begin{definition}
\label{def:cstarTightness}A set of positive measures $\mathbb{S}$ is said to
be uniformly $c^{\ast }$--tight if for every $\varepsilon >0$ there are $%
\varepsilon _{1},...,\varepsilon _{m}>0$ and a compact set $K_{\varepsilon
}\in T\times V$ such that $K_{\varepsilon }\subset \cup _{j}\left\{
c_{j}^{\ast }\geq \varepsilon _{j}\right\} $ and $S\left( K_{\varepsilon
}^{c}\right) <\varepsilon $, where the set $K_{\varepsilon }^{c}$ is the
complement of $K_{\varepsilon }^{{}}$ in $T\times V$. 
\end{definition}

\begin{example}
If $\mathbb{S}$ consists of measures supported on $\cup _{j}\left\{
c_{j}^{\ast }\geq \varepsilon \right\} $ for some $\varepsilon >0$, then $%
\mathbb{S}$ is uniformly $c^{\ast }$--tight.
\end{example}

By imposing the stronger assumption (yet natural for the problem we are 
studying) of $c^{\ast }$--tightness we sharpen the result in Lemma 
\ref{thm:CompactCondition} by having implication in both directions.

\begin{theorem}
\label{thm:CompactCondition2}Assume that $\mathbb{S}$ is uniformly $c^{\ast }
$--tight. Then $\overline{\mathbb{S}}$ has uniformly bounded total masses if
and only if $H_{c^{\ast }}\left( \overline{\mathbb{S}}\right) $ is a compact
subset of $\mathbb{R}_{+}^{m}$.
\end{theorem}

\begin{proof}
Clearly, if $\mathbb{S}$ has uniformly bounded total masses then $H_{c^{\ast
}}\left( \mathbb{S}\right) $ is bounded, since $c^{\ast }$is bounded and
continuous. If $\sup T_{c}^{\ast }\left( \mathbb{S}\right) =c$
(componentwise), then take $\varepsilon ,\varepsilon _{1},...,\varepsilon
_{m}>0$ and $K_{\varepsilon }$ such that $K_{\varepsilon }\subset \cup
_{j}\left\{ c_{j}^{\ast }\geq \varepsilon _{j}\right\} $ and $S\left(
K_{\varepsilon }^{c}\right) <\varepsilon $. Then for all $S\in \mathbb{S}$
we have $\sum_{j}\varepsilon _{j}S\left( K_{\varepsilon }\right) \leq
\sum_{j}\varepsilon _{j}S\left\{ c_{j}^{\ast }\geq \varepsilon _{j}\right\}
\leq \sum_{j}\int c_{j}^{\ast }dS\leq \sum_{j}c_{j}$ so the total mass of $S$
is $S\left( K_{\varepsilon }^{c}\right) +S\left( K_{\varepsilon }\right)
\leq \varepsilon +\left. \sum_{j}c_{j}\right/ \sum_{j}\varepsilon _{j}$.
\end{proof}

The next result is the main result of this section, not least from the 
point of view of applications. Any source $S$ can be approximated by 
a sequence of discrete sources $S_{j}$
(i.e., linear combination of instantaneous point sources), so it may not
come as a surprise that the set of measurements is related to the linear
combinations of values of $c^{\ast }$, which is the content of the following

\begin{theorem}
\label{thm:ConvexConicalHull}Assume that $K\subset T\times V$ is compact and
all $c_{j}^{\ast }\geq \varepsilon $ on $K$ for some $\varepsilon >0$, and
let $\mathbb{S}$ the set of all positive finite measures on $K$. Then $%
H_{c^{\ast }}\left( \mathbb{S}\right) $ is the closure of the convex conical
hull of $c^{\ast }\left( K\right) $.
\end{theorem}

\begin{proof}
$\mathbb{S}$ is a weakly closed set since $K$ is compact. Moreover, every $%
S\in \mathbb{S}$ is the weak limit of a sequence of discrete $S_{j}$
supported in $K$, i.e., $S_{j}=\sum_{k=1}^{N_{j}}c_{jk}\delta _{jk}$ where $%
\delta _{jk}$ are Dirac measures supported at suitable spacetime points $%
\left( t_{jk},x_{jk}\right) \in K$, and $c_{jk}>0$ and $%
\sum_{k=1}^{N_{j}}c_{jk}=\int dS$ for $k=1,...,N_{j}$ and $\ j=1,2,...$.
Also, $H_{c^{\ast }}\left( S_{j}\right) _{i}=\left\langle S_{j},c_{i}^{\ast
}\right\rangle =\sum_{k=1}^{N_{j}}c_{jk}c_{i}^{\ast }\left(
t_{jk},x_{jk}\right) $, so $H_{c^{\ast }}\left( S_{j}\right) $ is in the
conical hull of $c^{\ast }\left( K\right) $, and $H_{c^{\ast }}\left(
S_{j}\right) \rightarrow H_{c^{\ast }}\left( S\right) $ when $j\rightarrow
\infty $. This proves that $H_{c^{\ast }}\left( \mathbb{S}\right) $ is
included in the closure of the convex conical hull. Conversely, given a
point $y$ in the closure of the conical hull, there is a sequence $S_{j}$ of
discrete measures of the above form such that $H_{c^{\ast }}\left(
S_{j}\right) \rightarrow y$. Since all $c_{j}^{\ast }\geq \varepsilon $ on $K
$, the masses of the $S_{j}$'s must be uniformly bounded, and since they are
supported on the compact set $K$, they form a tight set of measures. By
Prohorov's theorem there is a subsequence $j_{k}\rightarrow \infty $ when $%
k\rightarrow \infty $ and a measure $S\in \mathbb{S}$ such that $%
S_{j_{k}}\rightarrow S$ weakly, and hence $H_{c^{\ast }}\left(
S_{j_{k}}\right) \rightarrow H_{c^{\ast }}\left( S\right) $ when $%
k\rightarrow \infty $. Hence $y=H_{c^{\ast }}\left( S\right) $, so $y\in
H_{c^{\ast }}\left( \mathbb{S}\right) $, which proves that the closure of
the convex conical hull is included in $H_{c^{\ast }}\left( \mathbb{S}%
\right) $.
\end{proof}

\section{Cones of measures}
In this section, as a preamble to the next section on the least squares solution, we give 
a technical lemma on the  closedness of cones generated by closed
bounded sets of measures.

For reach the desired result we have to introduce an additional condition on the 
generating set, namely a lower bound on the mass of $S$.

\begin{definition}
A set $\mathbb{S}$ of positive measures is said to have uniformly positive
total masses if there is a constant $M>0$ such that the total mass of $S$ is 
$\geq M$ for all $S\in \mathbb{S}$.
\end{definition}

\begin{lemma}
\label{thm:ClosedCone}Assume that $\mathbb{S}$ is a set of positive measures 
on $T\times V$, and let $\mathcal{%
C=}\operatorname{cone}\left( \mathbb{S}\right) $, the positive cone
generated by $\mathbb{S}$. Then $\operatorname{cone}\left( \overline{\mathbb{S}}%
\right) \subseteq $ $\overline{\mathcal{C}}$. Moreover,if $\mathbb{S}$ have
uniformly positive total masses, then $\operatorname{cone}\left( \overline{\mathbb{S}%
}\right) =\overline{\mathcal{C}}$.
\end{lemma}

\begin{proof}
The first statement follows from the fact that if $S_{j}\in \mathbb{S}$ and $%
S_{j}\rightarrow S$ weakly, then $\lambda S_{j}\rightarrow \lambda S$ for
all $\lambda \geq 0$. To prove the second statement, assume that $\mu \in 
\overline{\mathcal{C}}$, and take $\lambda _{j}S_{j}\in \mathcal{C}$ with $%
\lambda _{j}\geq 0$, $S_{j}\in \mathbb{S}$ and $\lambda _{j}S_{j}\rightarrow
\mu $ weakly. Since the $S_{j}$'s have uniformly bounded masses from below,
the $\lambda _{j}$'s are uniformly bounded, and hence there is a subsequence 
$j_{k}\rightarrow \infty $ such that $\lambda _{j_{k}}\rightarrow \lambda $
when $k\rightarrow \infty $. Hence $S_{j}\rightarrow \mu /\lambda $ weakly,
so $\mu /\lambda \in \overline{\mathbb{S}}$, i.e., $\mu \in \operatorname{cone}%
\left( \overline{\mathbb{S}}\right) $.
\end{proof}

The following example shows that the lower bound on the masses in $\mathbb{S}
$ is necessary for the second statement.

\begin{example}
Let $\mathbb{S}\mathcal{=}\left\{ S_{x}=x\delta _{x},x\in \left( 0,1\right)
\right\} $, a subset of all positive measures on $\mathbb{R}$. Then $%
\mathcal{C=}\left\{ \lambda x\delta _{x},x\in \left( 0,1\right) \text{ and }%
\lambda \geq 0\right\} $. Consider $\mu _{n}=nS_{1/n}=\delta _{1/n}\in 
\mathcal{C}$. Then $\mu _{n}\rightarrow \delta _{0}$ weakly so $\delta
_{0}\in \overline{\mathcal{C}}$. Suppose that $\delta _{0}\in \operatorname{cone}%
\left( \overline{\mathbb{S}}\right) $. Then $\lambda \delta _{0}\in 
\overline{\mathbb{S}}$ for some $\lambda >0$, so there is a sequence $%
x_{j}\downarrow 0$ such that $x_{j}\delta _{x_{j}}\rightarrow \lambda \delta
_{0}$ weakly. Hence $x_{j}f\left( x_{j}\right) \rightarrow \lambda f\left(
0\right) $ for all continuous functions, which is a contradiction since we
can have $f\left( 0\right) \neq 0$. We conclude that $\delta _{0}\notin 
\operatorname{cone}\left( \overline{\mathbb{S}}\right) $.
\end{example}

\section{Least squares solutions to inverse problems}

In addition to characterising the set of measurements, Theorem 
\ref{thm:ClosedCone} enables us to determine when the least squares 
inverse problem is well-defined (Theorem \ref{thm:LeastSquare} below). We 
begin by defining the least squares solution to the inverse problem.

\begin{definition}
\label{def:LeastSquare}Given adjoint plumes $c^{\ast }=\left( c_{1}^{\ast
},...,c_{m}^{\ast }\right) $ on $T\times V$, assumed continuous and bounded,
and given measurement values, $\bar{c}=\left( \bar{c}_{1},...,\bar{c}%
_{m}\right) $, and given a weakly closed cone $\mathcal{C}$ of positive
measures on $T\times V$, a least square solution to the inverse problem in $%
\mathcal{C}$ is a measure $\bar{S}\in \mathcal{C}$ such that%
\begin{equation}
\label{eqn:LeastSquares}
\left\Vert \bar{c}-H_{c^{\ast }}\left( \bar{S}\right) \right\Vert
=\min_{S\in \mathcal{C}}\left\Vert \bar{c}-H_{c^{\ast }}\left( S\right)
\right\Vert 
\end{equation}%
where $\left\Vert \cdot \right\Vert $ denotes the Euclidean norm in $\mathbb{%
R}^{m}$.
\end{definition}

Collecting the results from the previous sections we are now in a position to 
show when the least squares inverse problem is well defined. We assume 
that $\mathbb{S}$ has uniformly positive total masses and let $\mathcal{C=%
}\operatorname{cone}\left( \overline{\mathbb{S}}\right) $ where $\mathbb{S}$, 
then by Lemma \ref{thm:ClosedCone} it follows that $\mathcal{C}$ is weakly 
closed. We furthermore assume that $\mathbb{S}$ is tight and has uniformly 
bounded total masses which by Lemma \ref{thm:CompactCondition} implies 
that the image of the cone $\kappa=H_{c^{\ast }}\left( \mathcal{C}\right)$ 
is a closed positive cone in $\mathbb{R}_{+}^{m}$. We now prefer to 
conduct the analysis of the least squares inverse problem on the generating 
set alone and we therefore need the following lemma justifying that it 
suffices to solve a minimization problem on the generating set.

For a single ray, we have an analytical formula for the closest point,
namely, $\pi _{x}\left( z\right) =\left( x\cdot z\right) z/\left\Vert
z\right\Vert ^{2}$, the closest point from $x$ on the ray $\left\{ \lambda
z:\lambda >0\right\} $. Therefore we can minimize over a generating set
rather than over the full cone:

\begin{lemma}
\label{lem:ConicalMin}Assume that $\kappa $ is a positive cone in $\mathbb{R}_{+}^{m}$ generated
by a set $B\subset \mathbb{R}_{+}^{m}\setminus \left\{ 0\right\} $, and
assume that $x\in \mathbb{R}_{+}^{m}\setminus \kappa $, $x\neq 0$. If 
\begin{equation}
y\in \kappa \text{ and }\left\Vert y-x\right\Vert =\min_{w\in \kappa
}\left\Vert w-x\right\Vert   \label{eqn:ymin}
\end{equation}%
then there is a $z\in B$ such that $y=\pi _{x}\left( z\right) $, and 
\begin{equation}
\left\Vert \pi _{x}\left( z\right) -x\right\Vert =\min_{w\in B}\left\Vert
\pi _{x}\left( w\right) -x\right\Vert \text{.}  \label{eqn:zmin}
\end{equation}%
Conversely, if $z\in B$ satisfies (\ref{eqn:zmin}) then $y=\pi _{x}\left(
z\right) =\left\Vert \pi _{x}\left( z\right) \right\Vert z/\left\Vert
z\right\Vert $ satisfies (\ref{eqn:ymin}).
\end{lemma}

\begin{proof}
Minimizing over rays we have 
\[
\min_{w\in \kappa }\left\Vert w-x\right\Vert =\min_{w\in B}\left\Vert \pi
_{x}\left( w\right) -x\right\Vert 
\]
if either of the min exists. Moreover, for any $y\in \kappa $ there are $%
z\in B$ and $\lambda >0$ such that $y=\lambda z$, and $\pi _{x}\left(
y\right) =\pi _{x}\left( z\right) $. Consequently, for such $y$ and $z$, if
either $\left\Vert y-x\right\Vert =\min_{w\in \kappa }\left\Vert
w-x\right\Vert $ or $\left\Vert \pi _{x}\left( z\right) -x\right\Vert
=\min_{w\in B}\left\Vert \pi _{x}\left( w\right) -x\right\Vert $ holds we
have 
\[
\left\Vert y-x\right\Vert =\min_{w\in \kappa }\left\Vert w-x\right\Vert
=\min_{w\in \kappa }\left\Vert \pi _{x}\left( w\right) -x\right\Vert
=\min_{w\in B}\left\Vert \pi _{x}\left( w\right) -x\right\Vert =\left\Vert
\pi _{x}\left( z\right) -x\right\Vert 
\]
\end{proof}

In view of the previous lemma we see why we insisted on introducing 
the assumption on uniform positive total masses: it is important that 
the generating set does not contain the origin. Now, finally, we have 
come to the point where we can state, and easily prove, the main theorem:

\begin{theorem}
\label{thm:LeastSquare} Assume that the set of measures $\mathbb{S}$ is uniformly tight,
and weakly closed, with uniformly bounded and uniformly positive total
masses. Let $\mathcal{C}$ be the positive cone generated by $\mathbb{S}$.
Then $\mathbb{S}$ is weakly compact, and  $\mathcal{C}$ is weakly closed.
Moreover, there is a solution $\bar{S}$ to the least squares inverse problem
(\ref{eqn:LeastSquares}) on $\mathcal{C}$, given by%
\begin{equation}
\bar{S}=\frac{\left\Vert \pi _{\bar{c}}\left( H_{c^{\ast }}\left( \hat{S}%
\right) \right) \right\Vert }{\left\Vert H_{c^{\ast }}\left( \hat{S}\right)
\right\Vert }\hat{S}  \label{scaledSolution}
\end{equation}%
where $\hat{S}$ is a solution to the following least squares problem on $%
\mathbb{S}$:%
\begin{equation}
\left\Vert \pi _{\bar{c}}\left( H_{c^{\ast }}\left( \hat{S}\right) \right) -%
\bar{c}\right\Vert =\min_{S\in \mathbb{S}}\left\Vert \pi _{\bar{c}}\left(
H_{c^{\ast }}\left( S\right) \right) -\bar{c}\right\Vert   \label{reducedLS}
\end{equation}
\end{theorem}

\begin{proof}
The set $\mathbb{S}$ is weakly relatively compact by
Theorem 4 and hence weakly compact since it is assumed to be weakly closed.
The cone $\mathcal{C}$ generated by $\mathbb{S}$ is weakly closed by Theorem
12. The set $B=H_{c^{\ast }}\left( \mathbb{S}\right) $ is compact by Theorem
6, and $B\subset \mathbb{R}_{+}^{m}\setminus \left\{ 0\right\} $ since $%
\mathbb{S}$ has uniformly positive total masses. Let $\kappa =H_{c^{\ast
}}\left( \mathcal{C}\right) $. Then $\kappa $ is the positive cone generated
by $B$, and $\kappa $ is closed because the mapping $H_{c^{\ast }}$ is
continuous. Since $B$ is compact, there is a $z\in B$ such that $\left\Vert
\pi _{\bar{c}}\left( z\right) -\bar{c}\right\Vert =\min_{w\in B}\left\Vert
\pi _{\bar{c}}\left( w\right) -\bar{c}\right\Vert $. By the second statement
in Lemma 16, $y=\left\Vert \pi _{\bar{c}}\left( z\right) \right\Vert
z/\left\Vert z\right\Vert $ satisfies $\left\Vert y-\bar{c}\right\Vert
=\min_{w\in \kappa }\left\Vert w-\bar{c}\right\Vert $, and $y\in \kappa $
since $\kappa $ is closed. Finally, we take $\hat{S}\in \mathbb{S}$ such
that $H_{c^{\ast }}\left( \hat{S}\right) =z$; then $\hat{S}$ satisfies (\ref%
{reducedLS}) and $\bar{S}$ given by (\ref{scaledSolution}) has $H_{c^{\ast
}}\left( \bar{S}\right) =y$ and $\bar{S}$  is a solution to (\ref{eqn:LeastSquares}).
\end{proof}

Note that the solution is not necessarily unique, unless $\mathbb{S}$ is a
convex set of positive measures, in which case $\mathcal{C}$ is a closed
convex cone of positive measures and $\kappa=H_{c^{\ast }}\left( \mathcal{C}%
\right) $ is a closed convex cone in $\mathbb{R}_{+}^{m}$. Note also that it
suffices to find a minimizer in the generating set $\overline{\mathbb{S}}$,
and compute the scaling afterwards.

\begin{example}
Let $\mathbb{S}$ be the set of single instantaneous point sources in a
compact set $K\subset T\times V$. This is a uniformly tight, weakly closed
set of measures with uniformly bounded and uniformly positive total masses,
representing instantaneous point sources of unit mass. The positive cone $%
\mathcal{C}$ generated by $\mathbb{S}\mathcal{=}\overline{\mathbb{S}}$
represents all instantaneous point sources supported in $K$. Hence $%
B=H_{c^{\ast }}\left( \mathbb{S}\right) =c^{\ast }\left( K\right) $, the
image of $K$, is a basic set for the closed cone $\kappa=H_{c^{\ast }}\left( 
\mathcal{C}\right) $. Note that neither of the cones are convex;\ only
single instantaneous point sources, not linear combinations of different
ones, are included.
\end{example}

\begin{example}
Let $\mathbb{S}$ be the set of single continuous point sources with spatial
support in a compact set $K\subset V$ and unit total mass, i.e., 
\begin{equation}
S=q\left( t\right) dt\otimes \delta _{x^*}\left( dx\right) 
\end{equation}%
where $q$ is a nonnegative continuous function with $\int_{T}q\left(
t\right) dt=1$, and $x^*\in K$. Then the weak closure $\overline{\mathbb{S}}$
of $\mathbb{S}$ consists of all 
\begin{equation}
S=\mu \left( dt\right) \otimes \delta _{x^*}\left( dx\right) 
\end{equation}%
where $\mu $ is a probability measure on $T$. Note that $\overline{\mathbb{S}%
}$ includes temporally singular measures, for example discrete sums of
instantaneous point sources $S=\sum_{k}\lambda _{k}\delta _{t^*_{k}}\left(
dt\right) \otimes \delta _{x^*}\left( dx\right) $ with $\Sigma _{k}\lambda
_{k}=1$. This kind of singular measures must be included in order to obtain
a closed cone $H_{c^{\ast }}\left( \mathcal{C}\right) $, and thereby a
well--posed minimization problem.
\end{example}

\section{Conclusion}

We have presented a measure theoretic framework for studying the adjoint
dispersion problem. This framework and the accompanying measure theoretic
machinery enabled us to derive results for general linear inverse dispersion
problems without making prior assumptions on the number of sources, their
emission patterns and so on. Indeed, in our modus operandi the notion of
number of sources is not even a well-defined concept. We investigated when 
a given set of sensor
data can be realisable from a linear combination of source measures chosen
from some subset of all positive measures. Then we shifted the view from
working with a fixed set of measurement values, to asking (and answering)
the question: if the source is chosen from a closed cone of positive
measures, what are the possible measurement values that this source can
produce? Finally we used the framework to derive necessary and sufficient
conditions for the existence of a solution to the inverse least-squares
problem.

We conclude that the framework presented in this paper is a powerful tool
for stating and proving results on linear inverse atmospheric problems in their
full generality. The framework is not limited to proving the results that we
have presented here, indeed our next step is to use the framework to prove
rigorous results on the first order inverse method of Footprints, e.g. \cite%
{Robertson2004}, \cite{Pudykiewicz1998}. The framework is also easily
augmented to incorporate the forward dispersion problem as well. Our
preliminary investigations into uncertainty analysis of the forward
dispersion problem indicates that this is a fruitful approach.

\end{document}